\documentclass[letterpaper,11pt]{article}
\usepackage{fullpage}

\usepackage{ltexpprt}

\usepackage{comment}
\usepackage{amsmath,amssymb}
\usepackage{algorithm,algpseudocode}
\usepackage{cite}
\usepackage{enumerate}
\usepackage{bm}
\usepackage{cite}
\usepackage{hyperref}       % hyperlinks
\usepackage{tikz}
\usetikzlibrary{positioning,arrows,circuits.logic.US}

\newtheorem{definition}{Definition}[section]
\newtheorem{example}{Example}[section]
\newtheorem{remark}{Remark}[section]

\newcommand{\qed}{\hfill \ensuremath{\square}}
\newcommand{\blackqed}{\hfill \ensuremath{\blacksquare}}

\newcommand{\argmax}{\operatornamewithlimits{argmax}}

\newcommand{\OMIT}[1]{}

\makeatletter
\@addtoreset{equation}{section}
\makeatother

\usepackage{color}
\usepackage{CJKutf8}
\usepackage{ifthen}
\newcommand{\COMM}[2]{{
\begin{CJK}{UTF8}{ipxm}
\ifthenelse{\equal{#1}{MI}}{\color{blue}}{
\ifthenelse{\equal{#1}{TM}}{\color{red}}{
\ifthenelse{\equal{#1}{TR}}{\color{magenta}}{
\ifthenelse{\equal{#1}{BB}}{\color{cyan}}}}}
[#1: #2]
\end{CJK}
}}

\begin{document}

\title{Algorithmic Meta-Theorems for Monotone Submodular Maximization}
\author{
Masakazu Ishihata \\
NTT CS Laboratories \\
ishihata.masakazu@lab.ntt.co.jp\\
\and
Takanori Maehara 
\\ RIKEN Center for Advanced Intelligence Project \\ 
takanori.maehara@riken.jp
\and
Tomas Rigaux \\
\'Ecole Normale Sup\'erieure \\
tomas@rigaux.com
}
\date{}

\maketitle
% Copyright Statement
% When submitting your final paper to a SIAM proceedings, it is requested that you include 
% the appropriate copyright in the footer of the paper.  The copyright added should be 
% consistent with the copyright selected on the copyright form submitted with the paper.
% Please note that "20XX" should be changed to the year of the meeting.

% Default Copyright Statement
%\fancyfoot[R]{\footnotesize{\textbf{Copyright \textcopyright\ 2019 by SIAM\\
%Unauthorized reproduction of this article is prohibited}}}

% Depending on which copyright you agree to when you sign the copyright form, the copyright 
% can be changed to one of the following after commenting out the default copyright statement
% above.

%\fancyfoot[R]{\footnotesize{\textbf{Copyright \textcopyright\ 20XX\\
%Copyright for this paper is retained by authors}}}

%\fancyfoot[R]{\footnotesize{\textbf{Copyright \textcopyright\ 20XX\\
%Copyright retained by principal author's organization}}}

%\pagenumbering{arabic}
%\setcounter{page}{1}%Leave this line commented out.

%\input{abstract}
\begin{abstract}
We consider a monotone submodular maximization problem whose constraint is described by a logic formula on a graph.
Formally, we prove the following three ``algorithmic metatheorems.''

(1) If the constraint is specified by a monadic second-order logic on a graph of bounded treewidth, the problem is solved in $n^{O(1)}$ time with an approximation factor of $O(\log n)$.

(2) If the constraint is specified by a first-order logic on a graph of low degree, the problem is solved in $O(n^{1 + \epsilon})$ time for any $\epsilon > 0$ with an approximation factor of $2$.

(3) If the constraint is specified by a first-order logic on a graph of bounded expansion, the problem is solved in $n^{O(\log k)}$ time with an approximation factor of $O(\log k)$, where $k$ is the number of variables and $O(\cdot)$ suppresses only constants independent of $k$.
\end{abstract}

\clearpage
\tableofcontents
\clearpage

\section{Introduction}

\subsection{Problems and Results}

We consider \emph{monotone submodular maximization problems} whose feasible sets are subgraphs specified by \emph{monadic second-order (MSO) formula} and \emph{first-order (FO) formula} \footnote{ 
The \emph{first-order logic on graphs} is a language that consists of vertex variables $x$, edge predicates $e(x, y)$, and the usual predicate logic symbols ($\forall$, $\exists$, $\lnot$, $\land$, $\lor$, $=$, $\neq$, etc.)
The \emph{monadic second-order logic on graphs} extends the first-order logic by adding vertex subset variables $X$ and vertex inclusion predicate $x \in X$. 
In this paper, a \emph{first-order formula} is a formula expressed by the first-order logic on graphs.
A \emph{monadic second-order formula} is defined similarly.
%is a formula consisted of the same component of the first-order formula and vertex subset variables $X$, inclusion predicate $x \in X$, vertex variables $x$, edge predicates $e(x, y)$, and the standard logic symbols.
}.
 Formally, we consider the following two problems.

\begin{definition}[MSO-Constrained Monotone Submodular Maximization Problem]
\label{def:mso}
Let $G = (V(G), E(G))$ be an undirected graph, $\phi(X)$ be a monadic second-order formula with a free vertex-subset variable $X$% 
, and $f \colon 2^{V(G)} \to \mathbb{R}$ be a nonnegative monotone submodular function.%
\footnote{
For a finite set $V$, a function $f \colon 2^V \to \mathbb{R}$ is \emph{nonnegative} if $f(U) \ge 0$. 
$f$ is \emph{monotone} if $f(U) \le f(W)$ for all $U, W \subseteq V$ with $U \subseteq W$.
$f$ is \emph{submodular} if $f(U) + f(W) \ge f(U \cup W) + f(U \cap W)$ for all $U, W \subseteq V$.
}
Then, the MSO-constrained monotone submodular maximization problem is defined as follows.
\begin{align}
\begin{array}{ll}
\textrm{\rm maximize} & f(U) \\
\textrm{\rm subject to} & G \models \phi(U), \ U \subseteq V(G).\footnotemark
\end{array}
\end{align}
\footnotetext{$G \models \phi(U)$ means formula $\phi(X)$ is satisfied on $G$ when $U$ is substituted to the free variable $X$. $G \models \phi(u_1, \ldots, u_k)$ is defined similarly.}
\end{definition}

\begin{definition}[FO-Constrained Monotone Submodular Maximization Problem]
\label{def:fo}
Let $G = (V(G), E(G))$ be an undirected graph, $\phi(x_1, \ldots, x_k)$ be a first-order formula with free vertex variables $x_1, \ldots, x_k$% 
, and $f \colon 2^{V(G)} \to \mathbb{R}$ be a nonnegative monotone submodular function.
Then the \emph{FO-constrained monotone submodular maximization problem} is defined as follows.
\begin{align}
\begin{array}{ll}
\text{\rm maximize} & f(\{u_1, \ldots, u_k\}) \\
\text{\rm subject to} & G \models \phi(u_1, \ldots, u_k), \ u_1, \ldots, u_k \in V(G).
\end{array}
\end{align}
\end{definition}
In both the problems, we regard the length of the formula $|\phi|$ as a constant. 
In particular, in the FO-constrained problem, we regard the number $k$ of free variables (i.e., the cardinality of the solution) as a constant.

Both problems are very difficult, even for finding feasible solutions;
the MSO-constrained problem contains the three-coloring problem.
Therefore, unless P $=$ NP, we cannot obtain a feasible solution in polynomial time~\cite{karp1972reducibility}.
The FO-constrained problem can be solved in $O(n^k)$ time by an exhaustive search, where $n$ is the number of the vertices in the graph; however, it is difficult to improve this result, because the problem contains the $k$-clique problem, which cannot be solved in $n^{o(k)}$ time unless the exponential time hypothesis fails~\cite{chen2006strong}. 
Therefore, in both the problems, we have to restrict the graph classes suitably to obtain non-trivial results.
%On the other hand, if we restrict graph classes in these problems, we obtain feasible solutions in linear time or almost linear time according to the graph classes~\cite{grohe2017deciding}.

In this study, we show that the MSO- and FO-constrained monotone submodular maximization problems are well solved if the graphs are in certain classes as follows
(see Sections~\ref{sec:mso}, \ref{sec:fodeg}, and \ref{sec:foexp} for the definition of these graph classes).
Here, we assume that a submodular function is given by a value oracle, and is evaluated in $O(1)$ time.
\begin{theorem}
\label{thm:mso}
Let $\mathcal{G}$ be a class of graphs having bounded treewidth.
Then, for each $G \in \mathcal{G}$, the MSO-constrained monotone submodular maximization problem is solved in $n^{O(1)}$ time with an approximation factor of $O(\log n)$.%
\footnote{An algorithm has an approximation factor of $\alpha$ if $\alpha f(\text{ALG}) \ge f(\text{OPT})$ holds. where ALG is the solution obtained by the algorithm and OPT is the optimal solution.}
\blackqed
\end{theorem}
\begin{theorem}
\label{thm:fodeg}
Let $\mathcal{G}$ be a class of graphs having low degree.
Then, for each $G \in \mathcal{G}$, the FO-constrained monotone submodular maximization problem is solved in $O(n^{1 + \epsilon})$ time for any $\epsilon > 0$ with an approximation factor of $2$.
\blackqed
\end{theorem}
\begin{theorem}
\label{thm:foexp}
Let $\mathcal{G}$ be a class of graphs having bounded expansion.
Then, for each $G \in \mathcal{G}$, the FO-constrained monotone submodular maximization problem is solved in $n^{O(\log k)}$ time with an approximation factor of $O(\log k)$.
Here, $O(\cdot)$ suppresses only the constants independent of $k$.
\blackqed
\end{theorem}

\subsection{Background and Motivation}

\subsubsection{Submodular maximization.}

The problem of maximizing a monotone submodular function under some constraint is a fundamental combinatorial optimization problem, and has many applications in machine learning and data mining~\cite{krause2014submodular}.
This problem cannot be solved exactly in polynomially many function evaluations even for the cardinality constraint~\cite{feige1998threshold}; therefore, we consider approximation algorithms.

Under several constraints, the problem can be solved in polynomial time within a reasonable approximation factor. 
Examples include the cardinality constraint~\cite{nemhauser1978analysis}, knapsack constraint~\cite{sviridenko2004note}, and matroid constraint~\cite{calinescu2011maximizing}.
The problem is also solved on some graph-related constraints such as connectivity constraint~\cite{kuo2015maximizing} and $s$-$t$ path constraint~\cite{chekuri2005recursive}. 

Here, our research question is as follows:
\begin{quote}
\emph{What constraints admit efficient approximation algorithms for monotone submodular maximization problems}? 
\end{quote}

One solution to this question is given by Goemans et al.~\cite{goemans2009approximating}: 
If we can maximize linear functions on the constraint in polynomial time, the corresponding monotone submodular maximization problem can be solved in polynomial time with an approximation factor of $O(\sqrt{n} \log n)$.
This factor is nearly tight since we cannot obtain a $o(\sqrt{n} \log \log n / \log n)$ approximate solution in polynomially many oracle calls~\cite{bruggmann2017submodular}.
If a linear programming relaxation of the constraint has low correlation gap, we obtain an algorithm with an approximation factor that depends on the correlation gap by using the continuous greedy algorithm with the contention resolution scheme~\cite{vondrak2011submodular}.

In this study, we consider another approach. 
As in Definitions~\ref{def:mso} and \ref{def:fo}, we assume that the feasible sets are subgraphs of a graph specified by a logic formula.
To the best of our knowledge, no existing studies have considered this situation, and we believe that this situation is important in both practice and theory:
In practice, such problems appear in sensor network design problems~\cite{chekuri2005recursive,kuo2015maximizing}; thus, understanding classes of tractable problems helps practitioners to model problems. 
In theory, this may provide new algorithmic techniques because we need to combine quite different techniques in submodular maximization and mathematical logic. 

\subsubsection{Algorithmic metatheorem.}

In the field of algorithmic meatheorem, the constraints are represented by logic formulas~\cite{grohe2011methods}.
An \emph{algorithmic metatheorem} claims that if a problem is described in a certain logic and the inputs are structured in a certain way, then the problem can be solved with a certain amount of resources~\cite{tantau2016gentle}.
There are many existing algorithmic metatheorems, and Table~\ref{tbl:existing} shows some existing results.

% model checking 
The \emph{model checking problem on a graph $G$} asks whether the graph $G$ satisfies a certain property $\phi$, $G \models \phi$, or not. 
For the MSO formulas, Courcelle~\cite{courcelle1990monadic} showed that the model checking problem can be solved in linear time for bounded treewidth graphs. 
This result is tight for minor-closed graph classes~\cite{seese1991structure}. 
For the first-order formulas, Seese~\cite{seese1996linear} showed that the model checking problem can be solved in linear time for bounded degree graphs.
Later, this result was extended to low-degree graphs~\cite{grohe2001generalized}, bounded expansion graphs~\cite{grohe2017deciding} and nowhere dense graphs~\cite{grohe2018first}.

% counting
The \emph{counting problem on a graph $G$} asks the number of subgraphs which  satisfy a certain property $\phi$, that equals to the cardinality of the set $\mathcal{F} = \{ U \subseteq V(G) : G \models \phi(U) \}$ for the monadic second-order logic case and the set $\mathcal{F} = \{ (u_1, \ldots, u_k) \in V(G)^k : G \models \phi(u_1, \ldots, u_k) \}$ for the first-order case. 
The \emph{enumeration problem on a graph $G$} outputs the elements in $\mathcal{F}$ one-by-one.
As for the model checking problem, there are results on the monadic second-order logic with bounded treewidth graphs~\cite{bagan2006mso,amarilli2017circuit}, and the first-order logic with bounded degree graphs~\cite{kazana2011first}, low degree graphs~\cite{durand2014enumerating}, bounded expansion graphs~\cite{kazana2013enumeration}, and nowhere dense graphs~\cite{schweikardt2018enumeration}.

The \emph{linear maximization problem on a graph $G$} involves the maximization of a linear function on $\mathcal{F}$ defined above.
Compared with the model checking, counting, and enumeration problems, this problem is less studied.
There is a classical result on the monadic second-order logic with bounded treewidth graphs~\cite{arnborg1991easy}.
To the best of our knowledge, the result on the first-order logic with low degree graphs has not been explicitly stated yet. 
The result on the first-order logic with bounded expansion graphs has also not been explicitly stated; however, it is obtained by the same technique as that of Gajarsky et al.~\cite{gajarsky2017parameterized}, which shows the existence of a linear-sized extended formulation.
The possibility of extending these results for nowhere dense graphs is still open.

Our problems (Definitions~\ref{def:mso}, \ref{def:fo}) generalize linear functions to monotone submodular functions.
However, due to the submodularity, we need several new techniques to prove our theorems; see below.

\begin{table}[tb]
\caption{Existing results and our results on algorithmic metatheorems; the results without citations are shown in this paper. Each entry for model checking, counting, and linear maximization shows the time complexity, and each entry for submodular maximization column shows a pair of time complexity and approximation factor.}
\label{tbl:existing}
\centering
\begin{tabular}{cccccc}
Logic & Graph & Model Checking & Counting & Linear Max.  & Submod Max. \\ \hline 
MSO & Bounded Treewidth & $O(n)$~\cite{courcelle1990monadic} & $O(n)$~\cite{arnborg1991easy} & $O(n)$~\cite{arnborg1991easy} & ($n^{O(1)}$, $O(\log n)$) \\
%FO & Bounded Degree & $O(n)$~\cite{seese1996linear} & $O(n)$~\cite{kazana2011first} & $O(n)$~\cite{gajarsky2017parameterized} & \\
FO & Low Degree & $O(n^{1+\epsilon})$~\cite{grohe2001generalized} & $O(n^{1+\epsilon})$~\cite{durand2014enumerating} & $O(n^{1+\epsilon})$ & ($O(n^{1+\epsilon})$, $2$) \\
FO & Bounded Expansion & $O(n)$~\cite{dvovrak2013testing} & $O(n)$~\cite{kazana2013enumeration} & $O(n)$~\cite{gajarsky2017parameterized} & ($n^{O(\log k)}$, $O(\log k)$) \\
FO & Nowhere Dense & $O(n^{1+\epsilon})$~\cite{grohe2017deciding} & $O(n^{1+\epsilon})$~\cite{grohe2018first} & open & open \\
\end{tabular}
\end{table}

\subsection{Difficulty of Our Problems}

\paragraph{Difficulty of MSO-Constrained Problem on Bounded Treewidth Graphs.}

A linear function can be efficiently maximized on this setting~\cite{arnborg1991easy}. 
Therefore, it seems natural to extend their technique to the submodular setting.
However, we see that such an extension is difficult.

Their method first encodes a given graph as a binary tree by tree-decomposition.
Then, it converts a given monadic second-order formula into a tree automaton~\cite{thatcher1968generalized}.
Finally, it solves the problem using the bottom-up dynamic programming algorithm.
This gives the optimal solution in $O(n)$ time.

This technique cannot be extended to the submodular setting because a monotone submodular function cannot be maximized by the dynamic programming algorithm. 

\paragraph{Difficulty of FO-Constrained Problem on Low Degree Graphs.}

There are no existing studies on the linear maximization problem for low degree graphs.
Therefore, we need to establish a new technique.
In particular, in the result on the counting problems~\cite{durand2014enumerating}, they performed an inclusion-exclusion type algorithm. 
However, it is difficult to extend such a technique to the optimization problems.

\paragraph{Difficulty of FO-Constrained Problem on Bounded Expansion Graphs.}

To describe the difficulty of this case, we first introduce the algorithm for the linear maximization problem on bounded expansion graphs.%
\footnote{This is a simplified version of Gajarsky et al~\cite{gajarsky2017parameterized}'s proof for their result on extended formulations.}
If a graph class $\mathcal{G}$ has bounded expansion, then there exists functions $g \colon \mathbb{N} \to \mathbb{N}$ and $w\colon \mathbb{N} \to \mathbb{N}$ such that, for all $G \in \mathcal{G}$ and $k \in \mathbb{N}$, there is a coloring $c \colon V(G) \to \{1, \ldots, g(k)\}$ such that any $k$ colors induce a subgraph of treewidth bounded by $w(k)$.
Such coloring is referred to as low-treewidth coloring~\cite{nevsetvril2008grad1,nevsetvril2008grad2}.

The algorithm is described as follows.
First, we remove the universal quantifiers from the formula $\phi$ using Lemma~8.21 in \cite{grohe2011methods}.
Let $k' = k + l$, where $l$ is the number of existentially quantified variables.
Then, we find a low-treewidth coloring of $G$ with $g(k')$ colors.
Here, we can see that $k'$ colors are enough to cover all the variables in the formula.
Therefore, by solving the problems on all the $k'$ colored subgraphs using the algorithm for bounded treewidth graphs~\cite{arnborg1991easy}, we obtain the solution.

This technique cannot be extended to the submodular setting, because our result for the bounded treewidth graphs only gives an $n^{O(1)}$ time $O(\log n)$ approximation algorithm.
Since we can obtain the optimal solution in $O(n^k) = n^{O(1)}$ time through an exhaustive search, it does not make sense to reduce the problem to the bounded treewidth graphs.

In fact, most of the existing results for the first-order logic use the results for bounded treewidth graphs as a subroutine~\cite{dvovrak2013testing,grohe2017deciding,grohe2018first,kazana2013enumeration,schweikardt2018enumeration}.
However, the above discussion implies that we cannot use such a reduction for the submodular setting.

\subsection{Proof Outlines}

\subsubsection{Proof Outline of Theorem~\ref{thm:mso}.}

We represent the feasible set $\{ U \subseteq V(G) : G \models \phi(U) \} $ in the \emph{structured decomposable negation normal form (structured DNNF)}~\cite{darwiche2001decomposable} using Amarilli et al.~\cite{amarilli2017circuit}'s algorithm.
Here, a structured DNNF is a Boolean circuit based on the negation normal form, where the partition of variables is specified by a tree, called a \emph{vtree}.
%In theory, the structured DNNF is equivalent to the nice tree decomposition of the graph~\cite{kloks1994treewidth} with the tree automaton~\cite{comon2007tata}, where the tree of the tree decomposition corresponds to the vtree of the structured DNNF, and the states of the automaton corresponds to the logic gates in the structured DNNF~\cite{amarilli2017circuit}.

Then, we apply the \emph{recursive greedy algorithm}~\cite{chekuri2005recursive} to the structured DNNF.
We split the vtree at the centroid.
Then, we obtain constantly many subproblems whose numbers of variables are constant factors smaller than the original problem.
By solving these subproblems greedily and recursively, we obtain an $O(\log n)$-approximate solution in $n^{O(1)}$ time, since the recursion depth is $O(\log n)$ and the branching factor is $O(1)$.

\subsubsection{Proof Outline of Theorem~\ref{thm:fodeg}.}

By using Gaifman's locality theorem~\cite{gaifman1982local}, we decompose a given formula into multiple $r$-local formulas.
We perform the greedy algorithm with exhaustive search over the local formulas as follows:
First, we perform the exhaustive search to obtain the optimal solution for the first local formula in $O(n^{1 + \epsilon})$ time.
Then, by fixing the obtained solution, we proceed to the next local formula similarly.
By continuing this process until all the local formulas are processed, we obtain a solution.

In the above procedure, if each $r$-local part of the optimal solution are feasible to the corresponding subproblem, then the obtained solution is a $2$-approximate solution.
Otherwise, we can guess an entry of the optimal solution.
Thus, for each possibility, we call the procedure recursively. 
Then we obtain a recursion tree of size $O(n^{\epsilon})$.
We call this technique \emph{suspect-and-recurse}.
We show that there is at least one solution that has an approximation factor of $2$.

\subsubsection{Proof Outline of Theorem~\ref{thm:foexp}.}

We also use the suspect-and-recurse technique for this theorem; however, the tools used in each step are different.

By using the quantifier elimination procedure of Kazana and Segoufin~\cite{kazana2013enumeration}, we decompose a given formula into multiple ``tree'' formulas.
We perform the greedy algorithm with the recursive greedy algorithm over the tree formulas as follows.
First, we perform the \emph{recursive greedy algorithm} to obtain an $O(\log k)$-approximate solution to the first tree formula in $n^{O(\log k)}$ time.
Then, by fixing the obtained solution, we proceed to the next tree formula similarly.
By continuing this process until all the formulas are processed, we obtain a solution.

In the above procedure, if each tree part of the optimal solution is feasible to the corresponding subproblem, the obtained solution is an $O(\log k)$-approximate solution.
Otherwise, we can guess an entry of a \emph{forbidden pattern} that specifies which assignment makes the optimal solution infeasible.
For each possibility, we call the procedure recursively. 
Then, we obtain a recursion tree of size $O(1)$.
We show that there is at least one solution in the tree that has an approximation factor of $O(\log k)$.

\begin{comment}
\subsection*{Organization of the Paper}

The rest of the paper is organized as follows.
Section~\ref{sec:mso} gives the proof of Theorem~\ref{thm:mso}.
Section~\ref{sec:fodeg} gives the proof of Theorem~\ref{thm:fodeg}.
Section~\ref{sec:foexp} gives the proof of Theorem~\ref{thm:foexp}.
These proof uses different techniques; thus, we provide preliminaries in each section.
\end{comment}

%\section*{Notation}

%We use $u$ as a vertex, $x$ as a vertex variable, $\bar{x}$ as a set of vertex variables, and $X$ as a vertex subset variable.

%\clearpage

\section{Monadic Second-Order Logics on Bounded Treewidth Graphs}
\label{sec:mso}

\subsection{Preliminaries}

\subsubsection{Bounded Treewidth Graphs}

Let $G = (V(G), E(G))$ be a graph.
A \emph{tree decomposition} of $G$ is a tree $T = (V(T), E(T))$ with map $B \colon V(T) \to 2^{V(G)}$ satisfying the following three conditions~\cite{robertson1986graph}.
\begin{itemize}
\item $\bigcup_{t \in V(T)} B(t) = V(G)$.
\item For all $(u, v) \in E(G)$, there exists $t \in V(T)$ such that $u, v \in B(t)$.
\item For all $s, t \in V(T)$, $B(s) \cap B(t) \subseteq B(r)$ holds for all $r \in V(T)$ on $s$ and $t$.
\end{itemize}
The \emph{treewidth} of $G$ is given by $\min_T \max_{t \in V(T)} |B(t)| - 1$.
A graph class $\mathcal{G}$ has \emph{bounded treewidth} if there exists $w \in \mathbb{N}$ such that the treewidth of all $G \in \mathcal{G}$ is at most $w$.

\subsubsection{Structured Decomposable Negation Normal Form.}
\label{sec:dnnf}

% DNNF
%For any $n \in \mathbb{N}_+$, let $[n] := \{1,\dots,n\}$.
%For any $S \subseteq [n]$, we consider $\bm{b}_S = \{b_i \mid i \in S\}$ where $b_i$ is a Boolean variable, and $h$ is a Boolean function of $\bm{b}_S$.
%For any $S \subseteq V(G)$, we consider $\bm{b}_S = \{b_i \mid i \in S\}$ where $b_i$ is a Boolean variable, and $h$ is a Boolean function of $\bm{b}_S$.
The \emph{decomposable negation normal form (DNNF)} is a representation of a Boolean function~\cite{darwiche2001decomposable}.
Let $V$ be a finite set, and $\bar{b} = \{b_u \mid u \in V\}$ be a set of Boolean variables indexed by $V$.
Then, the DNNF is recursively defined as follows.
\begin{itemize}
\item  
The constants $\top$ (always true) and $\bot$ (always false) are in DNNF.
\item 
The literals $b_u$ and $\lnot b_u$ ($u \in V$) are in DNNF.
\item 
For any partition $\bar{b}_1, \bar{b}_2$ of variables $\bar{b}$ and formulas $h_1^{(i)}(\bar{b}_1)$ and $h_2^{(i)}(\bar{b}_2)$ ($i = 1, \dots, W$) in DNNF, the following formula is in DNNF.
\begin{align}
\label{eq:decomposable}
h(\bar{b}) = \bigvee_{i=1}^W h_1^{(i)}(\bar{b}_1) \land h_2^{(i)}(\bar{b}_2).
\end{align}
We call $h^{(i)}_1(\bar{b}_1)$ and $h^{(i)}_2(\bar{b}_2)$ \emph{factors} of this decomposition.
\end{itemize}
By recursively applying the above decomposition to each factor, every Boolean function can be represented as a DNNF~\cite{darwiche2001decomposable}.
The maximum number $W$ of disjunctions in \eqref{eq:decomposable}, which appears in the recursion, is called the \emph{width} of the DNNF.
A DNNF is usually represented by a Boolean circuit, which is a directed acyclic graph $\mathcal{D} = (V(\mathcal{D}), E(\mathcal{D}))$ whose internal gates are labeled ``AND'' or ``OR'', and the terminals are labeled $\top$, $\bot$, $b_i$, or $\lnot b_i$; see Example~\ref{ex:dnnf} below.

A \emph{vtree} $\mathcal{T} = (V(\mathcal{T}), E(\mathcal{T}))$ is a rooted full binary tree\footnote{A binary tree $\mathcal{T}$ is full if every non-leaf vertex has exactly two children.} whose leaves are the Boolean variables $\bar{b}$. 
A DNNF \emph{respects vtree $\mathcal{T}$} if, for any OR-gate of the DNNF, there exists an internal node $t \in V(\mathcal{T})$ of the vtree such that the partition $\bar{b}_1$ and $\bar{b}_2$ of variables of the decomposition represented by the OR-gate coincides with the leaves of the left and right subtrees of $t$.
A \emph{structured DNNF} is a DNNF that respects some vtree~\cite{darwiche2002knowledge}.

\begin{example}
\label{ex:dnnf}
This example is from Darwiche~\cite{darwiche2011sdd}.
Let $\bar{b} = \{b_1, b_2, b_3, b_4\}$, and $h(\bar{b}) = (b_1 \land b_2) \lor (b_2 \land b_3) \lor (b_3 \land b_4)$ be a Boolean formula.
We split $\bar{b}$ into $\{b_1, b_2\}$ and $\{b_3, b_4\}$. 
Then, the formula is factorized as
\begin{align}
h(\bar{b}) = \left( (b_1 \land b_2) \land \top \right) \lor \left(b_2 \land b_3\right) \lor \left( \top \land (b_3 \land b_4) \right).
\end{align}
This is in a structured DNNF. 
The circuit and the vtree are shown in Figures~\ref{fig:dnnf}, \ref{fig:vtree}.
The top OR-gate in Figures~\ref{fig:dnnf} represents the above decomposition and corresponds to the root node of the vtree shown in Figure~\ref{fig:vtree}.
\qed

\begin{figure}[tb]
\begin{minipage}{0.48\textwidth}
\centering
\tikzstyle{terminal}=[draw,circle,inner sep=0.2em]
\begin{tikzpicture}[circuit logic US,
                    line width=0.8pt,line cap=round,line join=round]
  
\node[or gate,logic gate inputs=nnn,rotate=90] at (0,0) (or1) {};

%\node[and gate,rotate=90] at (-1,-1.5) (and1) {};
%\node[terminal,label={[yshift=-2em]$u_2$}] at (-1.5,-2.5) (21) {};
%\node[terminal,label={[yshift=-2em]$\lnot u_3$}] at (-0.5,-2.5) (31) {};

%\draw 
%  (21) -- ++(0mm,1em) -| (and1.input 1)
%  (31) -- ++(0mm,1em) -| (and1.input 2);

\node[and gate,rotate=90] at (-2,-1.2) (and1) {};
\node[and gate,rotate=90] at (0,-1.2) (and2) {};
\node[and gate,rotate=90] at (2,-1.2) (and3) {};
\draw (or1.input 1) -- +(0mm,-2mm) -| (and1.output);
\draw (or1.input 2) -- +(0mm,-2mm) -| (and2.output);
\draw (or1.input 3) -- +(0mm,-2mm) -| (and3.output);

\node[or gate,rotate=90] at (-2.5,-2.4) (or2) {};
\node[terminal,label={[yshift=-2em]$\top$}] at (-1.5,-2.2) (top1) {};
\draw (and1.input 1) -- +(0mm,-2mm) -| (or2.output);
\draw (and1.input 2) -- +(0mm,-2mm) -| (top1);

\node[terminal,label={[yshift=-2em]$b_2$}] at (-0.5,-2.2) (z2) {};
\node[terminal,label={[yshift=-2em]$b_3$}] at (0.5,-2.2) (z3) {};
\draw (and2.input 1) -- +(0mm,-2mm) -| (z2);
\draw (and2.input 2) -- +(0mm,-2mm) -| (z3);

\node[terminal,label={[yshift=-2em]$\top$}] at (1.5,-2.2) (top2) {};
\node[or gate,rotate=90] at (2.5,-2.4) (or3) {};
\draw (and3.input 1) -- +(0mm,-2mm) -| (top2);
\draw (and3.input 2) -- +(0mm,-2mm) -| (or3.output);

\node[and gate,rotate=90] at (-2.6,-3.4) (and4) {};
\draw (or2.input 1) -- +(0mm,-2mm) -| (and4.output);

\node[terminal,label={[yshift=-2em]$b_1$}] at (-3.1,-4.4) (z22) {};
\node[terminal,label={[yshift=-2em]$b_2$}] at (-2.1,-4.4) (z32) {};
\draw (and4.input 1) -- +(0mm,-2mm) -| (z22);
\draw (and4.input 2) -- +(0mm,-2mm) -| (z32);

\node[and gate,rotate=90] at (2.4,-3.4) (and5) {};
\draw (or3.input 1) -- +(0mm,-2mm) -| (and5.output);

\node[terminal,label={[yshift=-2em]$b_3$}] at (1.9,-4.4) (z23) {};
\node[terminal,label={[yshift=-2em]$b_4$}] at (2.9,-4.4) (z33) {};
\draw (and5.input 1) -- +(0mm,-2mm) -| (z23);
\draw (and5.input 2) -- +(0mm,-2mm) -| (z33);
\end{tikzpicture}
\caption{Example of a structured DNNF.}
\label{fig:dnnf}
\end{minipage}
\begin{minipage}{0.48\textwidth}
\tikzstyle{node}=[circle, draw, thin, fill=white, inner sep=0.2em]
\centering
\begin{tikzpicture}[>=latex',node distance=4em]
    \node [node] at (0,0) (r) {};
    \node [node] at (-2,-1) (a) {};
    \node [node] at (2,-1) (b) {};
    \node [node,label={below:$b_1$}] at (-3,-2)   (c) {};
    \node [node,label={below:$b_2$}] at (-1.5,-2) (d) {};
    \node [node,label={below:$b_3$}] at (1,-2)    (e) {};
    \node [node,label={below:$b_4$}] at (3,-2)    (f) {};
    \path[-] (r) edge node {} (a);
    \path[-] (r) edge node {} (b);
    \path[-] (a) edge node {} (c);
    \path[-] (a) edge node {} (d);
    \path[-] (b) edge node {} (e);
    \path[-] (b) edge node {} (f);
\end{tikzpicture}
\caption{The vtree of the structured DNNF in Figure~\ref{fig:dnnf}.}
\label{fig:vtree}
\end{minipage}
\end{figure}
\end{example}

A boolean function $h$ can be used to represent a family of subsets of $V$.
We identify a subset $U \subseteq V$ as the indicator assignment $1_U$, which is defined by $b_u = 1$ ($u \in V)$ and $b_u = 0$ ($v \not \in V)$.
Then, $h$ represents a family of subsets $\{ U \subseteq V : h(1_U) = 1 \}$.
For simplicity, we say \emph{$U$ is in $\mathcal{D}$} if $h(1_U) = 1$ where $h$ is a Boolean function represented by $\mathcal{D}$.
Amarilli et al.~\cite{amarilli2017circuit} showed that a family of subsets in a bounded treewidth graph specified by a monadic second-order formula has a compact structured DNNF representation.
\begin{theorem}[Amarilli et al.~\cite{amarilli2017circuit}]
\label{thm:amarilli}
Let $G = (V(G), E(G))$ be a bounded treewidth graph, $\phi(X)$ be a monadic second-order formula, and $h$ is a Boolean function representing the family of subsets $\{ U \subseteq V(G) : G \models \phi(U) \}$.
Then, $h$ is represented by a structured DNNF with a bounded width.%
\footnote{Amarilli et al.~\cite{amarilli2017circuit} did not claim the structuredness and the boundedness of the DNNF. However, by observing their construction, these two properties are immediately confirmed.}
The structured DNNF is obtained in polynomial time.
\blackqed
\end{theorem}

\begin{remark}
\label{rem:linear}
The structured DNNF provides a ``syntax sugar'' of the Courcelle-type automaton technique~\cite{courcelle1990monadic}. 
Actually, the theorem is proved as follows. 
First, a bounded treewidth graph is encoded by a labeled binary tree using a tree-decomposition~\cite{bodlaender1996linear}.
Then, the given formula is interpreted as a formula on labeled trees, and is converted into a (top-down) tree automaton using a result of Thatcher and Wright~\cite{thatcher1968generalized}.
We consider the root vertex.
For each tree-automaton transition, we construct structured DNNFs for the subtrees.
Then, by joining the DNNFs with an AND-gate, and by joining the AND-gates with an OR-gate, we obtain the desired structured DNNF whose width is the number of states of the automaton.

Any proof with a structured DNNF can be converted into a proof using a tree-decomposition and a tree automaton by following the above construction.
However, in our case, the former approach gives a simpler proof than the latter one.
\qed
\end{remark}

\subsection{Proof of Theorem~\ref{thm:mso}}

We propose an algorithm to prove Theorem~\ref{thm:mso}.
First of all, we encode a given bounded treewidth graph and a monadic second-order formula into a structured DNNF $\mathcal{D} = (V(\mathcal{D}), E(\mathcal{D}))$ and the corresponding vtree $\mathcal{T} = (V(\mathcal{T}), E(\mathcal{T}))$ using Theorem~\ref{thm:amarilli}. 
Then, the problem is reduced to maximizing a monotone submodular function on $\mathcal{D}$.

Our algorithm is based on Chekuri and Pal's \emph{recursive greedy algorithm}~\cite{chekuri2005recursive}, which is originally proposed for $s$-$t$ path constrained monotone submodular maximization problem.
In this approach, we decompose the problem into several subproblems, and solve the subproblems one-by-one in a greedy manner.

We use leaf separators to obtain the subproblems.
An edge $e \in E(\mathcal{T})$ is a $\gamma$-\emph{leaf separator} if each subtree, which is obtained by removing $e$, has at most $\gamma$ fraction of leaves.
A full binary tree has a $(2/3)$-leaf separator as follows.
\begin{lemma}
\label{lem:separator}
A full binary tree $\mathcal{T} = (V(\mathcal{T}), E(\mathcal{T}))$ with $n \ge 2$ leaves has a $(2/3)$-leaf separator $e \in E(\mathcal{T})$.
Such $e$ is obtained in $O(n)$ time.
\end{lemma}
\begin{proof}
This proof is almost the same as that of Lemma~3 in \cite{xiao2013fptass}.
Let $e \in E(\mathcal{T})$ be the edge such that the difference between the numbers of leaves of the subtrees obtained by removing $e$ is the smallest.
Such $e$ is easily obtained in $O(n)$ time using a depth-first search.
We show that $e$ is a $(2/3)$-separator.

Let $\mathcal{T}_1$ and $\mathcal{T}_2$ be the subtrees obtained by removing $e$, and let $n_i$ be the number of leaves of $\mathcal{T}_i$.
Without loss of generality, we assume $n_1 < n_2$.
Let $\mathcal{T}_3$ and $\mathcal{T}_4$ be subtrees of $\mathcal{T}_2$ obtained by removing the other two edges adjacent to $e$.
By the definition of $n_1$, we have $n_1 \le n_3 + n_4$.
By considering the cut separating $T_3$ with the minimality of $e$, we have $n_3 + n_4 - n_1 \le |n_3 - n_4 - n_1|$.
Therefore, we have $n_3 \le n_1$.
Similarly, we have $n_4 \le n_1$.
Therefore, we have $n_1 \le n_3 + n_4 \le 2 n_1$. 
Since $n_1 + n_2 + n_3 = n$, we have $3 n_1 \le n$ and  $n_3 + n_4 \le 2n/3$.
This shows that $e$ is a $(2/3)$-separator.
\qed
\end{proof}

The subproblems are obtained as follows.
Let $(s,t) \in E(\mathcal{T})$ be a $(2/3)$-leaf separator of the vtree $\mathcal{T}$, where $t$ is a child of $s$.
Let $V_1$ be the vertices corresponding to the leaves of the subtree rooted by $t$, and let $V_2 = V \setminus V_1$.
By definition, there are OR-gates $\alpha_1, \ldots, \alpha_W \in V(\mathcal{D})$ that correspond to a factorization~\eqref{eq:decomposable} induced by $t$.
For each $j = 1, \ldots, W$, we define $\mathcal{D}^{(j)}_1$, which is the structured DNNF induced by the descendants (inclusive) of $\alpha_j$.
Then, $\mathcal{D}^{(j)}_1$ is a structured DNNF over $V_1$.
We also define $\mathcal{D}^{(j)}_2$ the structured DNNF obtained by replacing $\alpha_j$ to $\top$ and $\alpha_1, \ldots, \alpha_{j-1}, \alpha_{j+1}, \ldots, \alpha_N$ to $\bot$. 
Then, $\mathcal{D}^{(j)}_2$ is a structured DNNF over $V_2$.
Through this construction, we obtain the following two properties.
\begin{lemma}
\label{lem:widthpreserving}
The widths of $\mathcal{D}^{(j)}_1$, $\mathcal{D}^{(j)}_2$ ($j = 1, \ldots, W$) are at most that of $\mathcal{D}$.
\qed
\end{lemma}
\begin{lemma}
\label{lem:decompose}
Let $U_1 \subseteq V_1$ and $U_2 \subseteq V_2$.
Then $U_1 \cup U_2$ is in $\mathcal{D}$ if and only if there exists $j \in \{1, \ldots, W\}$ such that $U_1$ is in $\mathcal{D}^{(j)}_1$ and $U_2$ is in $\mathcal{D}^{(j)}_2$.
\qed
\end{lemma}
%\begin{proof}
%The formula is true only if at least one $\alpha_j$ is true.
%\end{proof}

By Lemma~\ref{lem:decompose}, for some $j = 1, \ldots, W$, the optimal solution $U^* \subseteq V(G)$ is partitioned to $U_1^*$ and $U_2^*$ that satisfy $\mathcal{D}_1^{(j)}$ and $\mathcal{D}_2^{(j)}$, respectively.
In the algorithm, we guess such $j$. 
Then, we solve the problem on $\mathcal{D}_1^{(j)}$ recursively to obtain $U_1^{(j)}$. 
Then, by modifying the function $f(X)$ to $f_{U_1^{(j)}}(U) = f(U_1^{(j)} \cup U) - f(U_1^{(j)})$, we solve the problem on $\mathcal{D}_2^{(j)}$ recursively to obtain $U_2^{(j)}$.
By taking the maximum over all $j$, we obtain a solution $U = U_1^{(j)} \cup U_2^{(j)}$ for some $j$.
The precise implementation is shown in Algorithm~\ref{alg:recursivegreedyMSO}.
We analyze this algorithm.

\begin{algorithm}[tb]
\caption{Recursive Greedy Algorithm on structured DNNF}
\label{alg:recursivegreedyMSO}
\begin{algorithmic}[1]
\Procedure{RecursiveGreedy}{$\mathcal{D}$, $f$}
\State{Compute an edge separator $(s, t) \in E(\mathcal{T})$, where $t$ is a child of $s$. }
\State{Let $\alpha_1, \ldots, \alpha_W$ be the AND-gates associated with $t$.}
\For{$j = 1, \ldots, W$}
\State{Construct $\mathcal{D}^{(j)}_1$ and $\mathcal{D}^{(j)}_2$.}
\State{Let $U^{(j)}_1 \leftarrow \textsc{RecursiveGreedy}(\mathcal{D}^{(j)}_1, f)$.}
\State{Let $U^{(j)}_2 \leftarrow \textsc{RecursiveGreedy}(\mathcal{D}^{(j)}_2, f_{U^{(j)}_1})$.}
\EndFor
  \State{\textbf{return} $U^{(j)}_1 \cup U^{(j)}_2$ that maximizes $f$.}
\EndProcedure
\end{algorithmic}
\end{algorithm}

\begin{lemma}
Algorithm~\ref{alg:recursivegreedyMSO} runs in $n^{O(1)}$ time with an approximation factor of $O(\log n)$.
\end{lemma}

\begin{proof}
First, we analyze the running time.
Since the number of vertices in each subproblem is at most $2/3$ to the original one, the depth of the recursion is $O(\log n)$.
The branching factor is $W = O(1)$, which is the width of the DNNF.
Here, we used Lemma~\ref{lem:widthpreserving}.
Therefore, the size of the tree is $O(1)^{\log n} = n^{O(1)}$.
Therefore, the running time is $n^{O(1)}$.

Next, we analyze the approximation factor.
Let $\alpha(n)$ be the approximation factor of the algorithm when it is applied to the problem on $n$ variables.
Let $U^*$ be the optimal solution. 
By Lemma~\ref{lem:decompose}, $U^*$ is partitioned into $U_1^*$ and $U_2^*$, and these are feasible to $\mathcal{D}_1^{(j)}$ and $\mathcal{D}_2^{(j)}$ for some $j$, respectively.
Let $U_1$ and $U_2$ be solutions returned at the $j$-th step.
Then, by the definition of $\alpha$, we have
\begin{align}
  \alpha(2n/3) f(U_1) &\ge f(U_1^*) \\
  \alpha(2n/3) (f(U_1 \cup U_2) - f(U_1)) &\ge f(U_1 \cup U_2^*) - f(U_1).
\end{align}
By adding these inequalities, we have
\begin{align}
\label{eq:msoapprox}
	\alpha(2n/3) f(U_1 \cup U_2) \ge f(U_1 \cup U_2^*) - f(U_1) + f(U_1^*).
\end{align}
We simplify the right-hand side of the above inequality.
Here, we prove a slightly general lemma, which will also be used in later sections.
\begin{lemma}
\label{lem:prefixdeviation}
Let $f \colon V \to \mathbb{R}$ be a nonnegative monotone submodular function, and let $U_1, \ldots, U_d, U_1', \ldots, U_d' \subseteq V$ be arbitrary subsets.
Then, the following inequality holds.
\begin{align}
\label{eq:prefixdeviation}
  \sum_{i = 1}^d \left( f(U_1 \cup \cdots \cup U_{i-1} \cup U_i') - f(U_1 \cup \cdots \cup U_{i-1}) \right) 
  \ge  f(U_1' \cup \cdots \cup U_d') - f(U_1 \cup \cdots \cup U_d).
\end{align}
\end{lemma}
\begin{proof}
We prove the lemma by induction.
If $d = 1$, the inequality is reduced to
\begin{align}
	f(U_1') \ge f(U_1') - f(U_1).
\end{align}
This holds since $f$ is a nonnegative function.
If $d \ge 2$, we have
\begin{align}
  & \sum_{i=1}^d \left( f(U_1 \cup \cdots \cup U_{i-1} \cup U_i') - f(U_1 \cup \cdots \cup U_{i-1}) \right) \notag \\ 
  &= f(U_1') + \sum_{i=2}^d \left( f_{U_1}(U_2 \cup \cdots \cup U_{i-1} \cup U_i') - f_{U_1}(U_2 \cup \cdots \cup U_{i-1}) \right),
\end{align}
where $f_{U_1}(U) = f(U_1 \cup U) - f(U_1)$.
Since $f_{U_1} \colon V \to \mathbb{R}$ is also a nonnegative monotone submodular function, we can use the inductive hypothesis as
\begin{align}
	(\text{LHS})  &\ge f(U_1') + f_{U_1}(U_2' \cup \cdots \cup U_d') - f_{U_1}(U_2 \cup \cdots \cup U_d) \notag \\
  &= f(U_1') + f(U_1 \cup U_2' \cup \cdots \cup U_d') - f(U_1 \cup \cdots \cup U_d) \notag \\
  &\ge^{(*)} f(U_1' \cap (U_1 \cup U_2' \cup \cdots \cup U_d')) + f(U_1 \cup U_1' \cup U_2' \cup \cdots \cup U_d') - f(U_1 \cup \cdots U_d) \notag \\
  &\ge^{(**)} f(U_1' \cup U_2' \cup \cdots \cup U_d') - f(U_1 \cup \cdots U_d),
\end{align}
where $(*)$ is the submodularity on the first two terms, and $(**)$ follows from the nonnegativity of the first term and monotonicity of the second term.
This proves the lemma.
\qed
\end{proof}

By using this lemma in \eqref{eq:msoapprox}, we obtain the inequality
\begin{align}
	(\alpha(2n/3) + 1) f(U_1 \cup U_2) \ge f(U_1^* \cup U_2^*),
\end{align}
which shows that the approximation factor of the algorithm satisfies the following recursion.
\begin{align}
	\alpha(n) \le \alpha(2n/3) + 1.
\end{align}
By solving this recursion, we have  $\alpha(n) = O(\log n)$.

\qed
\end{proof}
This lemma proves Theorem~\ref{thm:mso}.

%\clearpage

\section{First-Order Logics on Low Degree Graphs}
\label{sec:fodeg}

\subsection{Preliminaries}

In the first-order case, we work on tuples of variables and vertices.
To simplify the notation, we use $\bar{x}$ to represent a tuple $(x_1, \ldots, x_k)$ of variables and $\bar{u}$ to represent a tuple $(u_1, \ldots, u_k)$ of vertices.
We denote by $u_i \in \bar{u}$ to represent $u_i$ an element of $\bar{u}$.
For a set function $f \colon V(G) \to \mathbb{R}$ and a set of tuples $\bar{u}_1, \ldots, \bar{u}_K$, we define $f(\bar{u}_1, \ldots, \bar{u}_K) = f(\bigcup_{I = 1}^K \bigcup_{u \in \bar{u}_I } \{ u \})$.

\subsubsection{Low Degree Graphs.}

A graph class $\mathcal{G}$ has \emph{low degree} if for any $\epsilon > 0$, there exists $n_\epsilon \in \mathbb{N}$ such that for all $G \in \mathcal{G}$ with $n = |V(G)| \ge n_\epsilon$, the maximum degree of $G$ is at most $n^{\epsilon}$~\cite{grohe2011methods}.
A typical example of low degree graph class is the graphs of maximum degree of at most $(\log n)^c$ for some constant $c$.
The low degree graph class and the bounded expansion graph class, which we consider in the next section, are incomparable~\cite{grohe2011methods}.

\subsubsection{Gaifman's Locality Theorem}

Let $\mathrm{dist} \colon V(G) \times V(G) \to \mathbb{Z}$ be the shortest path distance between the vertices of $G$.
For a vertex $u \in V(G)$ and an integer $r \in \mathbb{Z}$, we denote by $N(u, r) = \{ v \in V(G) : \mathrm{dist}(u, v) \le r \}$ the ball of radius $r$ centered at $u$.
Also, for a tuple $\bar{u}$ of vertices and an integer $r \in \mathbb{Z}$, we define $N(\bar{u}, r) = \bigcup_{u \in \bar{u}} N(u, r)$.
For tuples $\bar{u}, \bar{u}'$ of vertices, we define $\mathrm{dist}(\bar{u}, \bar{u}') = \min_{u \in \bar{u}, u' \in \bar{u}'} \mathrm{dist}(u, u')$.
For variables $x$ and $y$, and integer $r \in \mathbb{Z}$, we denote by $\mathrm{dist}(x, y) \le r$ the first-order formula that represents the distance between $x$ and $y$ is less than or equal to $r$. 
For tuples $\bar{x}, \bar{x}'$ of variables, we denote by $\mathrm{dist}(\bar{x}, \bar{x}') \le r$ the formula $ \bigvee_{x \in \bar{x}, x' \in \bar{x}'} \mathrm{dist}(x, x') \le r$. 
Note that it is a first-order formula.

A first-order formula $\phi(x_1, \ldots, x_k)$ is \emph{$r$-local} if it satisfies the property
\begin{align}
	G \models \phi(\bar{u}) \iff G[N(\bar{u}, r)] \models \phi(\bar{u}), 
%	G \models \phi(u_1, \ldots, u_k) \iff G[N(u_1, r) \cup \cdots \cup N(u_k, r)] \models \phi(u_1, \ldots, u_k),
\end{align}
%where $\bar{u} = (u_1, \ldots, u_k)$, $N(\bar{u}, r) = N(u_1, r) \cup \cdots \cup N(u_k, r)$. 
where $G[\cdot]$ denotes the induced subgraph.
Intuitively, a formula is $r$-local it is determined by the $r$-neighborhood structure around the variables.

One of the most important theorems on the first-order logic is the Gaifman's locality theorem.
\begin{theorem}[Gaifman~\cite{gaifman1982local}]
Every first-order sentence\footnote{A \emph{sentence} is a formula without free variables.} $\phi$ is equivalent to a Boolean combination of sentences of the form
\begin{align}
	\exists y_1, \ldots, \exists y_l \left( \bigwedge_{i=1}^l \phi'(y_i) \land \bigwedge_{i \neq j} \mathrm{dist}(y_i, y_j) > 2r \right),
\end{align}
where $\phi'(y)$ is an $r$-local formula.
Furthermore, such a Boolean combination can be computed from $\phi$.
\blackqed
\end{theorem}

We can obtain the Gaifman's locality theorem \emph{for formulas} by considering all the partitions of the variables as follows.
\begin{lemma}[Equation (1) in Segoufin and Vigny~\cite{segoufin2017constant}]
\label{lem:gaifmanformula}
Every first-order formula $\phi(x_1, \ldots, x_k)$ is equivalent to a formula of the form
\begin{align}
\label{eq:gaifman}
	\bigvee_{\Pi = (\bar{x}_1, \ldots, \bar{x}_K): \text{partition of } \{x_1, \ldots, x_k\}} \left( \bigwedge_{I=1}^K \phi_I^\Pi(\bar{x}_I) \land D^\Pi(\bar{x}_1, \ldots, \bar{x}_K) \right),
\end{align}
where $\phi_1^\Pi(\bar{x}_1), \ldots, \phi_K^\Pi(\bar{x}_K)$ are $r$-local formulas, and $D^\Pi(\bar{x}_1, \ldots, \bar{x}_K)$ expresses the fact that $\mathrm{dist}(\bar{x}_I, \bar{x}_J) > 2 r$ for all $I \neq J$, and no refinement of $\Pi$ satisfies this property.
Furthermore, such a formula can be computed from $\phi(x_1, \ldots, x_k)$.
\blackqed
\end{lemma}

\subsection{Proof of Theorem~\ref{thm:fodeg}}

First of all, we transform the given formula into the form \eqref{eq:gaifman} in $O(1)$ time.
Then, we solve the problem for each disjunction of the formula.
By taking the maximum of the solutions for the disjunctions, we obtain a solution.
Thus, we now consider the case where formula $\phi(x_1, \ldots, x_k)$ is in the following form:
\begin{align}
\label{eq:gaifmanformula}
	\phi(x_1, \ldots, x_k) = \bigwedge_{I=1}^K \phi_I(\bar{x}_I) \land D(\bar{x}_1, \ldots, \bar{x}_K),
\end{align}
and the optimal solution $(\bar{u}_1^*, \ldots, \bar{u}_K^*)$ also satisfies this formula.

To design an approximation algorithm, we introduce the following concept.
A feasible solution $(\bar{u}_1, \ldots, \bar{u}_K)$ is an \emph{$\alpha$-prefix dominating solution} if 
\begin{align}
\label{eq:alphaprefixgreedy}
	\alpha \left( f(\bar{u}_1, \ldots, \bar{u}_I) - f(\bar{u}_1, \ldots, \bar{u}_{I-1}) \right) \ge 
	f(\bar{u}_1, \ldots, \bar{u}_I^*) - f(\bar{u}_1, \ldots, \bar{u}_{I-1})
\end{align}
for all $I = 1, \ldots, K$.
\begin{lemma}
\label{lem:prefixgreedy}
Let $(\bar{u}_1, \ldots, \bar{u}_K)$ be an $\alpha$-prefix dominating solution.
Then, it is an $\alpha+1$-approximate solution.
\end{lemma}
\begin{proof}
By adding \eqref{eq:alphaprefixgreedy} over $I = 1, \ldots, K$, we obtain
\begin{align}
\label{eq:alphaprefixgreedysum}
	\alpha f(\bar{u}_1, \ldots, \bar{u}_K) 
    \ge 
	\sum_{I=1}^K \left( f(\bar{u}_1, \ldots, \bar{u}_I^*) - f(\bar{u}_1, \ldots, \bar{u}_{I-1}) \right) 
    \ge^{(*)} f(\bar{u}_1^*, \ldots, \bar{u}_K^*) - f(\bar{u}_1, \ldots, \bar{u}_K).
\end{align}
Here, $(*)$ follows from Lemma~\ref{lem:prefixdeviation}.
By moving the right-most term to the left-hand side, we obtain the lemma.
\qed
\end{proof}

Lemma~\ref{lem:prefixgreedy} implies that we only have to construct a prefix dominating solution.
A natural approach will be a greedy algorithm.
Suppose that we have a partial solution $(\bar{u}_1, \ldots, \bar{u}_{I-1})$.
Then, we find a solution $\bar{u}_I$ for the $I$-th component by solving the subproblem.
Here, the exact $I$-th solution is efficiently obtained  as follows.
\begin{lemma}
\label{lem:exhaustivelowdeg}
For given $\bar{u}_1, \ldots, \bar{u}_{I-1}$, we find the exact solution $\bar{u}_I$ to $\max \{ f_{\bar{u}_1, \ldots, \bar{u}_{I-1}}(\bar{u}_I) : \bar{u}_1, \ldots, \bar{u}_I \text{ satisfies $\phi$ restricted on these variables} \}$ in $O(n^{1 + \epsilon r k})$ time.
\end{lemma}
\begin{proof}
We guess vertex $u_{i} \in V(G)$ that is assigned to the first variable of $\bar{x}_I$.
Then, by Lemma~\ref{lem:gaifmanformula}, all the other variables should be assigned by the vertices in the $r$-neighborhood of $u_i$.
The number of vertices in the $r$-neighborhood of $u_i$ is at most $n^{\epsilon r}$; thus, we can check all the assignments in $O(n^{\epsilon r k})$ time.
Therefore, we can check all the assignments in $O(n^{1 + \epsilon r k})$ time.
\qed
\end{proof}

By iterating this procedure, we obtain a (possibly partial) solution.
In this procedure, if the obtained solution is not a partial, and $I$-th component $\bar{u}_I^*$ of the optimal solution is feasible to the $I$-th subproblem for all $I$, the obtained solution is a $1$-prefix dominating solution, which is a $2$-approximate solution.
We call such a situation \emph{$\bar{u}^*$ is prefix feasible to $\bar{u}$}.
When is $\bar{u}^*$ not prefix feasible to $\bar{u}$?
By observing \eqref{eq:gaifman}, we can see that it is infeasible \emph{only if} the distance between $\bar{u}_1, \ldots, \bar{u}_{I-1}$ and $\bar{u}_I^*$ is less than or equal to $2 r$ for some $I$.
If we know $\bar{u}^*$, it is easy to avoid such a solution. 
Here, we develop a method to avoid such a solution without knowing $\bar{u}^*$.

Our idea is the following.
Suppose that we have a (possibly partial) solution $\bar{u}$ such that $\bar{u}^*$ is not prefix feasible to $\bar{u}$.
Then, there exist $u_i^* \in \bar{u}_I^*$ such that $\mathrm{dist}(\bar{u}_J, u_i^*) \le 2 r$ for some $J < I$.
This means that, at least one $u_i^*$ is in $N(\bar{u}, 2 r) = N(u_1, 2r) \cup \cdots \cup N(u_k, 2 r)$.
% HERE
Since the number of possibilities ($i$ and $u_i^*$) is $k^2 n^{2 \epsilon r}$, we can suspect it by calling the procedure recursively until we suspect $k$ assignment.
We call this technique \emph{suspect-and-recurse}.

The detailed implementation is shown in Algorithm~\ref{alg:greedylowdeg} which calls Algorithm~\ref{alg:suspectlowdeg} as a subroutine.
The algorithm maintains a current guess of some entries of optimal solution as a list $F \subseteq \mathbb{N} \times V(G)$, i.e., $(i, v) \in F$ means that we guess the $i$-th entry of the optimal solution is $v$.
Then, the feasible set to the $I$-th subproblem when the $J$-th solutions ($J < I$) and $F$ are specified is given by
\begin{align}
\mathcal{F}_I = \left\{ \bar{u}_I : \phi(\bar{u}_I) \wedge \bigwedge_{J < I} \mathrm{dist}(\bar{u}_J, \bar{u}_I) > 2r \wedge \bigwedge_{(i,v) \in F: x_j \not\in \bar{x}_I} \mathrm{dist}(\bar{u}_I, v) > 2r \right\}.
\end{align}

\begin{lemma}
Algorithm~\ref{alg:suspectlowdeg} runs in $O(n^{1 + 3 \epsilon r k})$ time with an approximation factor of $2$.
\end{lemma}
\begin{proof}
First, we analyze the running time.
The algorithm constructs a recursion tree, whose depth is $k$, and the branching factor is $k |N(\bar{u}, 2r)| \le k^2 n^{2 \epsilon r}$.
Thus, the size of the tree is at most $k^{2 k} n^{2 \epsilon r k}$.
In each recursion, the algorithm calls Algorithm~\ref{alg:greedylowdeg} that runs in $O(n^{1 + \epsilon k})$ time by Lemma~\ref{lem:exhaustivelowdeg} (with a modification to handle $F$).
Therefore, the total running time is $O(n^{1 + 3 \epsilon r k})$.
  
Next, we analyze the approximation factor.
By the above discussion, the algorithm seeks at least one solution $\bar{u}$ such that $\bar{u}^*$ is prefix feasible to $\bar{u}$.
Such a solution has an approximation factor of $2$ because of Lemma~\ref{lem:prefixgreedy}. 
Therefore, we obtain the lemma.
\qed
\end{proof}

We obtain Theorem~\ref{thm:fodeg} immediately from this lemma.
We replace $\epsilon$ by $\epsilon / 3 r k$. 
If $n \le n_\epsilon$, we solve the problem by an exhaustive search, which gives the exact solution in $O(1)$ time.
Otherwise, we apply Algorithm~\ref{alg:suspectlowdeg}.
This gives the desired result.

\begin{algorithm}[tb]
\caption{Greedy algorithm for low degree graphs.}
\label{alg:greedylowdeg}
\begin{algorithmic}[1]
\Procedure{GreedyLowDeg}{$F$}
\For{$I = 1, \ldots, K$}
\State{$\bar{u}_I \leftarrow \argmax \{ f(\bar{u}_1, \ldots, \bar{u}_{I-1}, \bar{u}_I) : \bar{u}_I \in \mathcal{F}_I \}$}
\EndFor
\State{\textbf{return} $\bar{u}_1, \ldots, \bar{u}_K$}
\EndProcedure
\end{algorithmic}
\end{algorithm}

\begin{algorithm}[tb]
\caption{Suspect-and-recurse algorithm for low degree graphs.}
\label{alg:suspectlowdeg}
\begin{algorithmic}[1]
\Procedure{SuspectRecurseLowDeg}{$F$}
\State{$\bar{u} \leftarrow \textsc{GreedyLowDeg}(F)$}
\If{$|F| < k$}
\For{$i = 1, \ldots, k$}
\For{$v \in N(\bar{u}, 2r)$}
\State{$\bar{u}_{i,v} \leftarrow \textsc{SuspectRecurseLowDeg}(F \cup \{ (i, v) \})$}
\EndFor
\EndFor
\EndIf
\State{\textbf{return} best solution among $\bar{u}$ and $\bar{u}_{i,v}$}
\EndProcedure
\end{algorithmic}
\end{algorithm}

\begin{remark}
If $f$ is a linear function, we can obtain the exact solution using Algorithm~\ref{alg:suspectlowdeg} as follows.
First, we enumerate all the possibilities that which variables take the same value. 
Then, for each possibility, we apply Algorithm~\ref{alg:suspectlowdeg} after removing the redundant variables.
We see that, if the variables in a $1$-prefix dominating solution $\bar{u} = (\bar{u}_1, \ldots, \bar{u}_K)$ are pairwise disjoint, instead of \eqref{eq:alphaprefixgreedysum}, the following inequality holds.
\begin{align}
	f(\bar{u}_1, \ldots, \bar{u}_K) 
    \ge 
	\sum_{I=1}^K \left( f(\bar{u}_1, \ldots, \bar{u}_I^*) - f(\bar{u}_1, \ldots, \bar{u}_{I-1}) \right) 
    = f(\bar{u}_1^*, \ldots, \bar{u}_K^*) 
\end{align}
This means that a $1$-prefix dominating solution is the optimal solution. 
Therefore, this procedure gives the optimal solution.
\qed
\end{remark}

%\clearpage

\section{First-Order Logics on Bounded Expansion Graphs}
\label{sec:foexp}

\subsection{Preliminaries}

\subsubsection{Bounded Expansion Graphs.}

Let $\vec{G} = (V(G), \vec{E}(G))$ be a directed graph. 
A \emph{$1$-transitive fraternal augmentation} is a minimal supergraph $\vec{H} = (V(H), \vec{E}(H))$ of $\vec{G}$ such that
\begin{description}
  \item[(transitivity)] if $(u, v) \in \vec{E}(G)$ and $(v, w) \in \vec{E}(G)$ then $(u, w) \in \vec{E}(H)$, and
  \item[(fraternality)] if $(u, v) \in \vec{E}(G)$ and $(u, w) \in \vec{E}(G)$ then at least $(v, w) \in \vec{E}(H)$ or $(w, v) \in \vec{E}(H)$
\end{description}
holds. By the minimality condition, we must have $V(H) = V(G)$.
A \emph{transitive fraternal augmentation} is a sequence of $1$-transitive fraternal augmentations $\vec{G} = \vec{G}_0 \subseteq \vec{G}_1 \subseteq \cdots$.
Note that a transitive fraternal augmentation is not determined uniquely due to the freedom of choice of the fraternal edges.

We say that a class $\mathcal{G}$ of graphs has \emph{bounded expansion}~\cite{nevsetvril2008grad1} if there exists a function $\Gamma' \colon \mathbb{N} \to \mathbb{N}$ such that, for each $G \in \mathcal{G}$, there exists an orientation $\vec{G}$ and a transitive fraternal augmentation $\vec{G} = \vec{G}_0 \subseteq \vec{G}_1 \subseteq \cdots$ where $\delta^-(\vec{G}_i) \le \Gamma'(i)$, where $\delta^-(\vec{G}_i)$ is the maximum in-degree of $\vec{G}_i$.%
\footnote{There are several equivalent definitions for bounded expansion. We choose the transitive fraternal augmentation for the definition because it is a kind of degree boundedness, and therefore, it looks similar to degree lowness.}
In a class of graphs having bounded expansion, we compute a suitable transitive fraternal augmentation efficiently as follows.

\begin{theorem}[Ne{\v{s}}et{\v{r}}il and Ossona de Mendez~\cite{nevsetvril2008grad2}]
\label{thm:nevsetvril2008grad2}
For a class $\mathcal{G}$ of graphs of bounded expansion, we can compute a transitive fraternal augmentation $\vec{G}_0 \subseteq \vec{G}_1 \subseteq \cdots \subseteq \vec{G}_i$ such that $\delta^-(\vec{G}_i) \le \Gamma(i)$ for some function $\Gamma \colon \mathbb{N} \to \mathbb{N}$.\footnote{$\Gamma(i)$ can be constant factor larger than the optimal $\Gamma'(i)$.}
\blackqed
\end{theorem}
Below, we fix the transitive fraternal augmentation computed by Theorem~\ref{thm:nevsetvril2008grad2}

\subsubsection{Kazana--Segoufin's Normal Form.}

Here, we introduce Kazana--Segoufin's normal form, which is proposed for the counting and enumeration problems for first-order formulas.

Let us consider the $i$-th graph $\vec{G}_i$ in the transitive fraternal augmentation.
We can represent the graph structure by $\Gamma(i)$ functions $\rho_1, \ldots, \rho_{\Gamma(i)} \colon V(G) \to V(G)$ such that $\rho_p(u)$ represents the $p$-th adjacent vertex of $u$.
For simplicity, we define $\rho_0(u) = u$ for all $u \in V(\vec{G}_i)$.
Now, we take a $1$-transitive fraternal augmentation $\vec{G}_{i+1}$ of $\vec{G}_i$. 
Since $\vec{G}_{i+1}$ is a supergraph of $\vec{G}_i$, and the in-degrees of $\vec{G}_{i+1}$ are bounded by $\Gamma(i+1)$, we also represent the graph structure by $\Gamma(i+1)$ functions $\rho_1, \ldots, \rho_{\Gamma(i+1)}$, where $\rho_1, \ldots, \rho_{\Gamma(i)}$ represent the graph structure of $\vec{G}_i$.
Now, we store the information of newly added edges as follows.
Suppose that the edge $(u, w) \in E(\vec{G}_{i+1})$ is added because of the transitivity of $(u, v) \in E(\vec{G}_i)$ and $(v, w) \in E(\vec{G}_i)$.
Let $\rho_p(w) = v$, $\rho_q(v) = u$, and $u = \rho_r(w)$.
Then we add the label ``$\rho_r = \rho_q \circ \rho_p$'' to vertex $w$. 
Similarly, suppose that edge $(u, w) \in E(\vec{G}_{i+1})$ is added because of the fraternality of $(u, v) \in E(\vec{G}_i)$ and $(w, v) \in E(\vec{G}_i)$.
Let $\rho_p(w) = v$, $\rho_q(u) = v$, and $u = \rho_r(w)$.
Then we add the label ``$\rho_r = \rho_q \circ \rho_p^{-1}$'' to vertex $w$. 
The number of labels required to represent all the relations is $\Gamma(i+1) \Gamma(i)^2$.

By Gaifman's locality theorem (Theorem~\ref{lem:gaifmanformula}), any first-order formula $\phi(\bar{x})$ is $r$-local for some constant $r \in \mathbb{N}$.
Therefore, by taking an $r$-transitive fraternal augmentation, we can eliminate all the quantifiers from the formula.
Kazana and Segoufin~\cite{kazana2013enumeration} showed that the quantifier eliminated formula is in the following form.
\begin{theorem}[Kazana and Segoufin~\cite{kazana2013enumeration}]
\label{thm:kazana2013enumeration}
Any first-order formula $\phi(x_1, \ldots, x_k)$ on a class of graphs having bounded expansion is equivalent to the quantifier free formula
\begin{align}
\label{eq:kazana}
  \phi(x_1, \ldots, x_k) = \bigvee_{\nu} \phi_\nu(x_1, \ldots, x_{k-1}) \land \tau_\nu(x_k) \land \Delta^=_\nu(x_1, \ldots, x_k) \land \Delta^{\neq}_\nu(x_1, \ldots, x_k),
\end{align}
where $\nu$ runs over $O(1)$ disjunctions, $\tau_\nu(x_k)$ is a formula containing constantly many labels on $x_k$ and terms of the form $f(g(x_k)) = h(x_k)$, $\Delta^=_\nu(x_1, \ldots, x_k)$ consists of at most one term of the form $\rho_p(x_j) = \rho_q(x_k)$, and $\Delta^{\neq}_\nu(x_1, \ldots, x_k)$ consists of constantly many terms of the form $\rho_p(x_j) \neq \rho_q(x_k)$.
Such a formula is obtained in $O(1)$ time.
\blackqed
\end{theorem}

\subsection{Proof of Theorem~\ref{thm:foexp}}

The proof of Theorem~\ref{thm:foexp} follows a similar strategy to the proof of Theorem~\ref{thm:fodeg}.
First, we solve the problem.
Then, we guess some information about the optimal solution.
By constructing a recursion tree of bounded size, we find at least one good solution in the tree.

First of all, we convert the given formula into a tractable form.
By expanding Kazana and Segoufin's normal form~\eqref{eq:kazana}, we obtain the following result.
\begin{lemma} %[Equality Forest Lemma]
\label{lem:normalform}
Any first-order formula $\phi(x_1, \ldots, x_k)$ on a class of graphs having bounded expansion is equivalent to
\begin{align}
\label{eq:normalform}
\phi(x_1, \ldots, x_k) = \bigvee_{\bar \nu} \tau_{\bar \nu}(x_1, \ldots, x_k) \land \Delta^=_{\bar \nu}(x_1, \ldots, x_k) \land \Delta^{\neq}_{\bar \nu}(x_1, \ldots, x_k)
\end{align}
  where $\tau_{\bar \nu}(x_1, \ldots, x_k) = \tau_{\nu_1}(x_1) \land \cdots \land \tau_{\nu_k}(x_k)$ is a formula that depends on the labels of $x_1, \ldots, x_k$ independently, $\Delta^=_{\bar \nu}(x_1, \ldots, x_k)$ is a conjunction of terms of the form $\rho_p(x_i) = \rho_q(x_j)$ such that the \emph{equality graph $\mathcal{X} _{\bar{\nu}} = (V(\mathcal{X}_{\bar{\nu}}), E(\mathcal{X}_{\bar{\nu}}))$ of $\Delta^{=}_{\bar{\nu}}(x_1, \ldots, x_k)$}, which is defined by $V(\mathcal{X}_{\bar{\nu}}) = \{ x_1, \ldots, x_k \}$ and $E(\mathcal{X}_{\bar{\nu}}) = \{ (x_i, x_j) : \text{``$\rho_p(x_i) = \rho_q(x_j)$''} \in \Delta_{\bar{\nu}}^{=}(x_1, \ldots, x_k) \}$, forms a forest, and  $\Delta^{\neq}_{\bar \nu}(x_1, \ldots, x_k)$ consists of constantly many terms of the form $\rho_p(x_i) \neq \rho_q(x_j)$.
Such a formula is obtained in $O(1)$ time.
\end{lemma}
\begin{proof}
By expanding $\phi(x_1, \ldots, x_k)$ using Theorem~\ref{thm:kazana2013enumeration} recursively, we obtain \eqref{eq:normalform}. 
In each disjunction, each variable $x_j$ appears in at most one equality constraint $\rho_p(x_i) = \rho_q(x_j)$ for $i < j$. 
Therefore, it forms a forest.
\qed
\end{proof}

As in the proof in Theorem~\ref{thm:fodeg}, without loss of generality, we assume that the given formula is in the form
\begin{align}
	\phi(x_1, \ldots, x_k) = \tau(x_1, \ldots, x_k) \land \Delta^{=}(x_1, \ldots, x_k) \land \Delta^{\neq}(x_1, \ldots, x_k),
\end{align}
and the optimal solution $\bar{u}^*$ also satisfies this formula.
Let $\mathcal{Y}_1 = (V(\mathcal{Y}_1), E(\mathcal{Y}_1)), \ldots, \mathcal{Y}_K = (V(\mathcal{Y}_K), E(\mathcal{Y}_K))$ be the connected components (trees) of the equality graph $\mathcal{X}$.
Then, the solution is decomposed into $\bar{u}_1, \ldots, \bar{u}_K$, where $\bar{u}_J$ is the variables of $J$-th tree $\mathcal{Y}_J$.
We try to construct a prefix dominating solution.

\medskip

\subsubsection{Simpler Case: No Inequality Constraints.}

First, we consider the case where $\Delta^{\neq}(x_1,\ldots,x_k)$ is empty. 
Suppose that we have a partial solution $(\bar{u}_1, \ldots, \bar{u}_{J-1})$ and try to find a solution $\bar{u}_J$ to the $J$-th subproblem.
Here, an $O(\log k)$-approximate solution is obtained by the recursive greedy algorithm as follows.

\begin{lemma}
\label{lem:simpler}
For given $\bar{u}_1, \ldots, \bar{u}_{I-1}$, we find an $O(\log k)$-approximate solution $\bar{u}_I$ to $\max \{ f_{\bar{u}_1, \ldots, \bar{u}_{I-1}}(\bar{u}_I) : \bar{u}_1, \ldots, \bar{u}_I \text{ satisfies $\phi$ restricted on these variables} \} $ in $n^{O(\log k)}$ time.
\end{lemma}
\begin{proof}
For notational simplicity, we omit subscript, and suppose that the variables of $\mathcal{Y}$ are $x_1, \ldots, x_k$.
First, we find a \emph{centroid} $x \in V(\mathcal{Y})$ of the tree.
Here, a centroid is a node such that all the subtrees $\mathcal{Y}^{(1)}, \ldots, \mathcal{Y}^{(d)}$ of $\mathcal{Y} \setminus \{x\}$ have at most $k/2$ vertices.
We guess the assignment $u \in V(G)$ of $x$, and modify the function by $f_u$.
We call the procedure recursively to the first subtree $\mathcal{Y}^{(1)}$ to obtain a solution $\bar{u}^{(1)}$ to $\mathcal{Y}^{(1)}$ that is consistent with the the assignment of $u$.
Then, we modify the function by $f_{u,\bar{u}^{(1)}}$ and call the procedure recursively to find a solution $\bar{u}^{(2)}$ of $\mathcal{Y}^{(2)}$, where the assignment $\bar{y}^{(2)} \leftarrow \bar{u}^{(2)}$ should be consistent with the assignments $y \leftarrow u$ and $\bar{y}^{(1)} \leftarrow \bar{u}^{(2)}$. 
By continuing this procedure, we obtain a sequence of solutions $u, \bar{u}^{(1)}, \ldots, \bar{u}^{(d)}$. 
The solution to the previous trees and this tree is given by $\bar{u} = (u, \bar{u}^{(1)}, \ldots, \bar{u}^{(d)})$.
The detailed implementation is shown in Algorithm~\ref{alg:recursivebddexp}.
The algorithm maintains a partial assignment as a list $A \subseteq \mathbb{N} \times V(G)$ such that $(i, u_i) \in A$ implies $u_i$ is assigned to $x_i$.

First, we analyze the running time. 
The algorithm constructs a recursion tree. 
Since the depth is $\log k$ and the branching factor is $n$, the size is $n^{O(\log k)}$.
In each recursion, the complexity is polynomial. 
Therefore, the running time is $n^{O(\log k)}$.

Next, we analyze the approximation factor.
Let $\bar{u}^*$ be the optimal solution, and $u^*$ be a component of $\bar{u}^*$ that is assigned to $x$.
The algorithm tries all the assignments; thus, we consider a step when $u^*$ is assigned to $x$.
Let $\bar u^{(1)}, \ldots, \bar u^{(d)}$ be the solutions obtained by calling the algorithm recursively.
Let $\alpha(k)$ be the approximation factor of the algorithm.
Then, we have
\begin{align}
  &\alpha(k/2) \left( f(u^*, \bar{u}^{(1)}, \ldots, \bar{u}^{(i-1)}, \bar{u}^{(i)}) - f(u^*, \bar{u}^{(1)}, \ldots, \bar{u}^{(i-1)}) \right) \notag \\
  &\ge \left( f(u^*, \bar{u}^{(1)}, \ldots, \bar{u}^{(i-1)}, \bar{u}^{(i)*}) - f(u^*, \bar{u}^{(1)}, \ldots, \bar{u}^{(i-1)}) \right).
\end{align}
This implies that $(u^*, \bar{u}^{(1)}, \ldots, \bar{u}^{(d)})$ forms an $\alpha(k/2)$-prefix dominating solution.
Therefore, by Lemma~\ref{lem:prefixgreedy}, it gives an $(\alpha(k/2) + 1)$-approximate solution.
Hence, the approximation factor of the algorithm satisfies $\alpha(k) \le \alpha(k/2) + 1 \le \cdots \le \log k$. 
Therefore, it gives a solution with an approximation factor of $O(\log k)$.
\qed
\end{proof}

\begin{algorithm}[tb]
\caption{Recursive greedy algorithm for bounded expansion graph (no inequality constraints).}
\label{alg:recursivebddexp}
\begin{algorithmic}
  \Procedure{RecursiveGreedyBddExp'}{$\mathcal{Y}$, $A$}
  \State{Find a centroid $x_i \in V(\mathcal{Y})$}
  \State{Let $\mathcal{Y}^{(1)}, \ldots, \mathcal{Y}^{(d)}$ be the trees of $\mathcal{Y} \setminus x_i$}
  \For{$u \in V(G)$}
  \If{the assignment $(i, u)$ is consistent with $A$}
  \State{$A(u) = A \cup \{ (i, u) \}$}
  \For{$j = 1, \ldots, d$}
  \State{$A(u) \leftarrow \textsc{RecursiveGreedyBddExp'}(\mathcal{Y}^{(j)}, A(u))$}
  \EndFor
  \EndIf
  \EndFor
  \State{\textbf{return} the best assignment among $A(u)$}
\EndProcedure
\end{algorithmic}
\end{algorithm}

By iterating this procedure, we obtain a solution $(\bar{u}_1, \ldots, \bar{u}_K)$.
The procedure is shown in Algorithm~\ref{alg:greedybddexp}.
Since there are no constraints between the variables in different trees, 
the optimal solution $\bar{u}^*$ is prefix feasible to $\bar{u}$; therefore, $\bar{u}$ is a $O(\log k)$-prefix dominating solution, which is a $O(\log k)$-approximate solution.

\begin{algorithm}[tb]
\caption{Greedy Algorithm for bounded expansion graphs (no inequality constraints).}
\label{alg:greedybddexp}
\begin{algorithmic}[1]
\Procedure{GreedyBddExp'}{\,}
\State{$A = \emptyset$}
\For{$I = 1, \ldots, K$}
\State{$A \leftarrow \textsc{RecursiveGreedyBddExp'}(\mathcal{Y}_I, A)$}
\EndFor
\State{\textbf{return} the solution $\bar{u}$ corresponds to assignment $A$}
\EndProcedure
\end{algorithmic}
\end{algorithm}

\subsubsection{General Case.}

Next, we extend the procedure when the formula contains inequalities.
Suppose that we obtain a solution $(\bar{u}_1, \ldots, \bar{u}_K)$ by calling the algorithm.
If $\bar{u}^*$ is prefix feasible to $\bar{u}$, it is a desired solution.
We try to extract some information when $\bar{u}^*$ is not prefix feasible.

By the normal form~\eqref{eq:normalform}, the interactions between the different trees only come from the inequalities.
More precisely, the solution $\bar{u}_J^*$ is infeasible to the $J$-th subproblem only if there is some $u_i \in \bar{u}_{I}$ ($I \leq J$) and $u_j^* \in \bar{u}_J^*$ such that $\rho_p(u_i) \neq \rho_q(u_j^*)$ fails for some $\text{``$\rho_p(x_i) \neq \rho_q(x_j)$''} \in \Delta^{\neq}(x_1, \ldots, x_k)$.
If we know $\bar{u}^*$, it is easy to avoid such a solution. 
Thus, we develop a suspect-and-recurse type algorithm.

In the proof of Theorem~\ref{thm:fodeg}, we directly guessed an entry of $\bar{u}^*$; however, in this case, it is impossible because the number of candidates $\{ v \in V(G) : \rho_p(u_i) = \rho_q(v) \}$ can be $\Omega(n)$.
To overcome this issue, we guess an entry of the \emph{forbidden pattern}.
A forbidden pattern is a set of tuples $F \subseteq \mathbb{N} \times \mathbb{N} \times V(G)$, and a solution $(u_1, \ldots, u_k)$ \emph{satisfies the forbidden pattern} $F$ if $\rho_p(u_i) \neq v$ for all $(i, p, v) \in F$.
Then, $\bar{u}^*$ is prefix feasible to $\bar{u}$ if and only if $\bar{u}$ satisfies the forbidden pattern
\begin{align}
	F^* = \{ (i, p, \rho_q(u_j^*)) \in \mathbb{N} \times \mathbb{N} \times V(G): \text{``$\rho_p(x_i) \neq \rho_q(x_j)$''} \in \Delta^{\neq}(x_1, \ldots, x_k) \}.
\end{align}
The size of this forbidden pattern is $|F^*| = |\Delta^{\neq}(x_1, \ldots, x_k)| = O(1)$.
If $\bar{u}^*$ is not prefix feasible to $\bar{u}$, we can guess one entry of the forbidden pattern since it must contain $(i, p, \rho_p(u_i))$ for some $i$ and $p$.
Therefore, we can implement the suspect-and-recurse type algorithm for $F$.

The detailed implementation is shown in Algorithm~\ref{alg:suspectbddexp1}.

\begin{lemma}
Algorithm~\ref{alg:suspectbddexp1} runs in $n^{O(\log k)}$ time with an approximation factor of $O(\log k)$.
\end{lemma}
\begin{proof}
First, we analyze the running time. 
The algorithm constructs a recursion tree.
Since the depth and the branching factor are $|\Delta^{\neq}(x_1, \ldots, x_k)| = O(1)$, the size of the tree is $O(1)$.
In each recursion, the algorithm calls Algorithm~\ref{alg:greedybddexp1}, whose running time is $n^{O(\log k)}$ by the same analysis as Lemma~\ref{lem:simpler}.
Therefore, the running time of Algorithm~\ref{alg:suspectbddexp1} is $n^{O(\log k)}$.

Next, we analyze the approximation factor.
If the optimal solution is prefix feasible to the current solution $\bar{u}$, the approximation factor is obtained by the same analysis as Lemma~\ref{lem:simpler}.
Otherwise, by the construction, it suspects at least one entry in $F^*$. 
By enumerating all possibilities up to the size of $F^*$, we find at least one solution that makes $\bar{u}^*$ prefix feasible in the tree.
\qed
\end{proof}
This lemma proves Theorem~\ref{thm:foexp}.

\begin{algorithm}[tb]
\caption{Recursive greedy algorithm for bounded expansion graphs.}
\label{alg:recursivebddexp1}
\begin{algorithmic}
  \Procedure{RecursiveGreedyBddExp}{$\mathcal{Y}$, $A$, $F$}
  \State{Find a centroid $x_i \in V(\mathcal{Y})$}
  \State{Let $\mathcal{Y}^{(1)}, \ldots, \mathcal{Y}^{(d)}$ be the trees of $\mathcal{Y} \setminus x_i$}
  \For{$u \in V(G)$}
  \If{the assignment $(i, u)$ is consistent with $A$ and satisfies $F$}
  \State{$A(u) = A \cup \{ (i, u) \}$}
  \For{$j = 1, \ldots, d$}
  \State{$A(u) \leftarrow \textsc{RecursiveGreedyBddExp}(\mathcal{Y}^{(j)}, A(u), F)$}
  \EndFor
  \EndIf
  \EndFor
  \State{\textbf{return} the best assignment among $A(u)$}
\EndProcedure
\end{algorithmic}
\end{algorithm}

\begin{algorithm}[tb]
\caption{Greedy Algorithm for bounded expansion graphs.}
\label{alg:greedybddexp1}
\begin{algorithmic}[1]
\Procedure{GreedyBddExp}{$F$}
%\For{$I = 1, \ldots, K$}
%\State{Find an $O(\log k)$-approximate solution $\bar{u}_I$ to the $I$-th subproblem by Algorithm~\ref{alg:recursivebddexp1}}
\State{$A = \emptyset$}
\For{$I = 1, \ldots, K$}
\State{$A \leftarrow \textsc{RecursiveGreedyBddExp}(\mathcal{Y}_I, A, F)$}
\EndFor
%\State{\textbf{return} $\bar{u}_1, \ldots, \bar{u}_K$}
\State{\textbf{return} the solution $\bar{u}$ corresponds to assignment $A$}
\EndProcedure
\end{algorithmic}
\end{algorithm}

\begin{algorithm}[tb]
\caption{Suspect-and-recurse algorithm for bounded expansion graphs.}
\label{alg:suspectbddexp1}
\begin{algorithmic}[1]
\Procedure{SuspectRecurseBddExp}{$F$}
\State{$\bar{u} = (u_1, \ldots, u_k) \leftarrow \textsc{GreedyBddExp}(F)$}
\If{$|F| < |\Delta^{\neq}(x_1, \ldots, x_k)|$}
\For{$\text{``$\rho_p(x_i) \neq \rho_q(x_j)$''} \in \Delta^{\neq}(x_1, \ldots, x_k)$}
\State{$\bar{u}_{i,p,j,q} \leftarrow \textsc{SuspectRecurseBddExp}(F \cup \{ (i, p, \rho_q(u_i)) \}$}
\EndFor
\EndIf
\State{\textbf{return} the best solution among $\bar{u}$ and $\bar{u}_{i,p,j,q}$}
\EndProcedure
\end{algorithmic}
\end{algorithm}

\section*{Acknowledgment}

We thank Antoine Amarilli for confirming the structuredness of their construction of DNNF in \cite{amarilli2017circuit}. 
We thank Jakub Gajarsky for the bibliography of the first-order logic constrained linear maximization problem.

%\clearpage

\bibliographystyle{siamplain}
\bibliography{main}

\end{document}